\documentclass{article}
\usepackage{spconf}

\usepackage{cite}
\usepackage{url}
\usepackage{textcomp}

\usepackage{array}
\usepackage{multirow}
\usepackage{arydshln}

\usepackage{amsmath,amssymb,amsfonts}
\usepackage{stfloats}
\usepackage{nicefrac}

\usepackage{graphicx}
\usepackage[caption=false, font=footnotesize]{subfig}
\usepackage{xcolor}
\usepackage{tikz,pgfplots}
\usetikzlibrary{arrows}

\usepackage[acronym]{glossaries}

\usepackage{algorithm}

\usepackage{caption}

\usepackage{amsmath}
\usepackage{mathtools}

\usepackage{enumitem}
\usepackage{array}
\usepackage{amsthm}
\usepackage[titletoc,toc,title]{appendix}
\usepackage[bookmarksopen=true]{hyperref}
\usepackage[capitalise,nameinlink]{cleveref}

\usepackage{bookmark}

\usepackage{csquotes}
\usepackage{url}
\usepackage{mathrsfs}

\usepackage{algorithmic}
\usepackage{scrextend}

\graphicspath{Figures/}
\newcommand{\tr}{^\top}			
\newcommand{\bs}[1]{\boldsymbol{#1}}	
\allowdisplaybreaks[4]					
\long\def\comment#1{}

\newfont{\bbb}{msbm10 scaled 700}

\newcommand{\uv}{{\bf u}}

\newcommand{\xv}{{\bf x}}
\newcommand{\yv}{{\bf y}}

\newcommand{\Dm}{{\bf D}}

\newcommand{\Id}{{\bf I}}
\newcommand{\Jm}{{\bf J}}

\newcommand{\Mm}{{\bf M}}

\newcommand{\Qm}{{\bf Q}}
\newcommand{\Rm}{{\bf R}}
\newcommand{\Sm}{{\bf S}}

\newcommand{\Um}{{\bf U}}
\newcommand{\Wm}{{\bf W}}

\newcommand{\Ec}{{\cal E}}

\newcommand{\Gc}{{\cal G}}
\newcommand{\Hc}{{\cal H}}
\newcommand{\Ic}{{\cal I}}

\newcommand{\Kc}{{\cal K}}
\newcommand{\Lc}{{\cal L}}
\newcommand{\Mc}{{\cal M}}
\newcommand{\Nc}{{\cal N}}

\newcommand{\Pc}{{\cal P}}

\newcommand{\Sc}{{\cal S}}

\newcommand{\Vc}{{\cal V}}

\newcommand{\Scc}{{\Sc^c}}

\newcommand{\Lscr}{ {\mathscr L}}
\newcommand{\Mscr}{ {\mathscr M}}
\newcommand{\Hscr}{ {\mathscr H}}

\newcommand{\Ex}{{\mathbb{E}}}

\newtheorem{theorem}{Theorem}

\newtheorem{prop}{Proposition}

\theoremstyle{remark}
\newtheorem{rem}{Remark}
\crefname{prop}{Proposition}{Propositions}

\usepackage{nccmath}

\newacronym{fir}{FIR}{finite-extent impulse response}
\newacronym{iir}{IIR}{infinite-extent impulse response}
\newacronym{psnr}{PSNR}{peak-signal-to-noise ratio}
\newacronym{snr}{SNR}{signal-to-noise ratio}
\newacronym{awgn}{AWGN}{additive white Gaussian noise}

\newacronym{nrmse}{NRMSE}{normalized-root-mean-square error}
\newacronym{cnn}{CNN}{convolutional neural network}
\newacronym{iid}{iid}{independent and identically distributed}
\newacronym{dft}{DFT}{discrete Fourier transform}
\newacronym{idft}{IDFT}{inverse discrete Fourier transform}
\newacronym{gft}{GFT}{graph Fourier transform}
\newacronym{igft}{IGFT}{inverse graph Fourier transform}
\newacronym{fft}{FFT}{fast Fourier transform}
\newacronym{fpga}{FPGA}{field-programmable gate array}

\newacronym{lp}{LP}{linear programming}
\newacronym{cpa}{CPA}{Chebyshev polynomial approximation}
\newacronym{ls}{LS}{least square}
\newacronym{arma}{ARMA}{autoregressive moving average}
\newacronym{wls}{WLS}{weighted least-square}
\newacronym{socp}{SOCP}{second-order cone programming}

\newacronym{sse}{SSE}{sum-of-square-error}

\newacronym{gsp}{GSP}{graph signal processing}

\def\BibTeX{{\rm B\kern-.05em{\sc i\kern-.025em b}\kern-.08em
    T\kern-.1667em\lower.7ex\hbox{E}\kern-.125emX}}
\title{Graph-based Signal Sampling with Adaptive Subspace Reconstruction for Spatially-irregular Sensor Data}
\name{Darukeesan Pakiyarajah, Eduardo Pavez, Antonio Ortega}
\address{University of Southern California, Los Angeles, CA, USA}
\begin{document}

\ninept
\maketitle
\begin{abstract}
Choosing an appropriate frequency definition and norm is critical in graph signal sampling and reconstruction. Most previous works define frequencies based on the spectral properties of the graph and use the same frequency definition and $\ell_2$-norm for optimization for all sampling sets. Our previous work demonstrated that using a sampling set-adaptive norm and frequency definition can address challenges in classical bandlimited approximation, particularly with model mismatches and irregularly distributed data. 
In this work, we propose a method for selecting sampling sets tailored to the sampling set adaptive GFT-based interpolation. When the graph models the inverse covariance of the data, we show that this adaptive GFT enables localizing the bandlimited model mismatch error to high frequencies, and the spectral folding property allows us to track this error in reconstruction. Based on this, we propose a sampling set selection algorithm to minimize the worst-case bandlimited model mismatch error. 
We consider partitioning the sensors in a sensor network sampling a continuous spatial process as an application. 
Our experiments show that sampling and reconstruction using sampling set adaptive GFT significantly outperform methods that used fixed GFTs and bandwidth-based criterion.
\end{abstract}
\vspace{-1ex}
\section{INTRODUCTION}
\vspace{-1ex}
\label{sec:intro}
\glsresetall

Graph signal sampling and interpolation involve observing a signal at $s$ of the $N$ graph nodes and using graph models to produce an interpolated signal at the remaining $N-s$ nodes~\cite{tanaka20}. 
Since the sampled signal lies in an $s$-dimensional space,  when using linear interpolators, the reconstructed signal belongs to an $s$-dimensional 
subspace of the original $N$-dimensional space. 
Existing works typically fix this subspace \textit{a priori} (e.g., the span of the first $s$ basis vectors of the GFT) for all sampling sets~\cite{anis16,Puy18,chen15,Tanaka19,Chamon18,Narang13,Chen15c,Chen16,Xie17}. 
This raises the question of how to optimize sampling strategies when a suitable subspace prior is not available for the data. 
While in standard graph models signals that are smooth with respect to the graph (low-frequency signals) are assumed to be more likely, this does not mean that a subspace approximation (assuming only the first $s$ frequencies are non-zero) is desirable. 
In particular, in 
our previous work~\cite{Paki24}, we showed that even if the original signal is smooth, the signal obtained by introducing zeros in unobserved nodes has potentially large energy in arbitrary frequencies, so a simple low-pass filter reconstruction produces poor results.  
As an alternative, we proposed designing a sampling set-dependent GFTs and corresponding reconstruction subspaces for which low pass filtering yields robust reconstruction, outperforming standard approaches~\cite{anis16,Puy18,chen15,Tanaka19,Chamon18,Narang13,Chen15c,Chen16,Xie17,Yang21}. 
While in \cite{Paki24}, we assumed the sampling set was given, here we address sampling set selection. 

Our work is based on the irregularity-aware graph Fourier transform ~\cite{girault18}, which allows for selecting an inner product as an additional degree of freedom to define different frequency representations of signals. In~\cite{girault20}, the authors proposed using an alternate norm (independent of the sampling set) for sampling and reconstruction that can take into account the irregularities in a sensor network capturing a spatially regular signal. In their experiments, they observed that different norms and corresponding frequency definitions could yield different reconstruction performances for different sampling sets~\cite{girault20}. In our previous work \cite{Paki24}, we introduced a sampling set-adaptive norm that can account for irregularities: 1)  the relative contribution of sampled nodes for interpolation,  and 2) the relative reliability of interpolated nodes. We showed that this interpolation can yield a robust reconstruction compared to the other methods that use standard sampling set-independent norm for optimization, particularly in the presence of model mismatches.

In this work, we assume that the graph Laplacian approximates the inverse covariance of the data and that the signals have a decaying power spectrum. We use the interpolator from \cite{Paki24}, which is based on a reconstruction subspace that is a function of the sampling set and is not fixed a priori. We show that this interpolation error depends only on the higher frequency content of the original signal as a function of the sampling set. Based on this, we compute the covariance of energy in the high frequencies and formulate the sampling set selection problem as an E-optimal design~\cite{anis16}. Furthermore, we show that the optimal linear Gauss Markov random field (GMRF) estimator can be interpreted as a subspace reconstruction with the sampling-set-adaptive frequency definition. From this, we justify how our reconstruction method can mitigate the effects of model mismatches for non-GMRF data.

\begin{figure} 
        \centering
        \subfloat[\label{fig:training}]{%
       \includegraphics[width = 0.5 \textwidth, trim={0cm 0.65cm 3cm 0.5cm}, clip ]{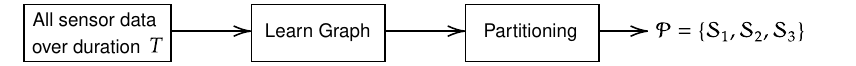}}
        \hfill \vskip -0.2ex
      \subfloat[\label{fig:implementation}]{%
            \includegraphics[width = 0.5 \textwidth, trim={0cm 0.5cm 3cm 0.1cm}, clip ]{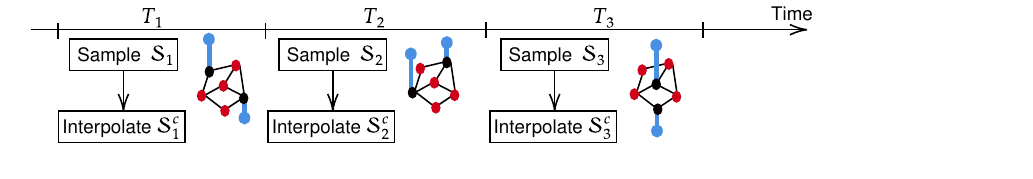}}
        \vskip -2ex
        \caption{(a) Graph learning and partitioning state; (b) Different subsets of sensors are  sampled at each time window. Blue and red vertices correspond to sampled and interpolated sensors, respectively.}        
        \vskip -4.5ex       
\end{figure}

As a use case, we address the problem of finding multiple sampling subsets for sensor networks sampling a spatial process~\cite{Chakraborty23,Holm21,li2024finding}. First, we collect data from all sensors over a duration of $T$ and learn a graph model. In this scenario, the learned graph may not capture the statistics of the underlying data (which was produced by a continuous process independent of the sensor locations) and provide only an approximation of the model. We then identify multiple disjoint sampling subsets as illustrated in~\cref{fig:training}. In practice, these subsets alternatively sample during different time frames and the readings of non-sampled sensors are interpolated, as shown in~\cref{fig:implementation}. This minimizes sensor active time and extends the sensor lifetime. We compare the interpolation performance to the state-of-the-art bandlimited reconstruction method~\cite{Chakraborty23}, showing that our approach outperforms the state-of-the-art by approximately $4$ dB when interpolating signals varying rapidly in space with fewer samples.

\vspace{-1ex}
\section{Preliminaries}
\label{sec:review}

\subsection{Graph signals and $(\Mm, \Qm)$-GFTs}

Let $\Gc=(\Vc,\Ec)$ denote a  undirected graph, where $\mathcal{V}=\{1,2,\cdots,N\}$ is the vertex set, and $\Ec \subset \Vc \times \Vc$ is the edge set. A graph signal is a function $x: \Vc \to \mathbb{R}$, with vector representation  given by $\xv = [x_1, x_2, \hdots, x_N] \tr$, where $x_i$ is the signal value at vertex $i \in \Vc$. The graphs are represented by a positive semi-definite variation operator $\Mm=(m_{ij})$, where $m_{ij}$ is the weight of the edge between vertex $i$ and $j$, and $m_{ij}=0$, whenever $(i,j)\notin \Ec$. The variation operator measures signal smoothness in a Hilbert space where the inner product is defined as $\langle \xv, \yv \rangle_\Qm =  \yv\tr\Qm \xv$ with induced $\Qm$-norm $\|\xv\|_\Qm^2 = \langle \xv, \xv \rangle_\Qm$ where $\Qm \succ 0$. The $(\Mm, \Qm)$-GFT \cite{girault18} is defined as the generalized eigenvectors of $\Mm$ which solves $\Mm \uv_k = \lambda_k \Qm \uv_k$,
where $0 \leq \lambda_1 \leq \hdots \leq \lambda_N$ are the graph frequencies and the columns of $\Um = [\uv_1, \hdots \uv_N]$ forms a $\Qm$-orthonormal (i.e., $\Um\tr\Qm\Um=\Id$) basis for $\mathbb{R}^N$. This corresponds to the decomposition $\Qm^{-1}\Mm = \Um\bs\Lambda\Um\tr\Qm$, where $\bs\Lambda=\text{diag}(\lambda_1,\hdots,\lambda_N)$. The $(\Mm, \Qm)$-GFT representation of a signal $\xv$ is defined as $\hat\xv=\Um\tr\Qm\xv$ and the corresponding inverse transform is given by $\xv = \Um\hat\xv$.
\vspace{-1.5ex}
\subsection{Bandlimited graph signal interpolation}
Given a signal $\xv$  and a sampling set $\Sc\subseteq\Vc$, the bandlimited reconstruction aims to interpolate $\xv$  from $\xv_\Sc$ in the form $\xv = \sum_{k=1}^K\hat x_k\uv_k = \Um_{\Vc\Kc}\hat\xv_\Kc$, where $\Kc=\{1,\hdots,K\}$ and $K\leq |\Sc|$.
The sample consistent bandlimited interpolation is obtained as~\cite{girault20}
\begin{align}
    \tilde{\xv} & = \Um_{\Vc\Kc}\left( \text{arg }\underset{\hat\xv_\Kc}{\text{min }} \| \xv_\Sc - \Um_{\Sc\Kc}\hat\xv_\Kc\|_{\Qm_\Sc}^2\right) \notag \\
    & = \Um_{\Vc\Kc}\left(\Um_{\Sc\Kc}\tr\Qm_\Sc\Um_{\Sc\Kc}\right)^{-1}\Um_{\Sc\Kc}\tr\Qm_\Sc\xv_\Sc. \label{eq:mqrec}
\end{align}
The sampling set selection methods aim to find the optimal sampling set that minimizes the effect of model mismatch and/or noise~\cite{tanaka20,girault20}.

In our previous work~\cite{Paki24}, we introduced a sampling set-adaptive interpolation with a specific choice of $\Qm$, which results in graph frequencies exhibiting a property known as spectral folding (SF)\cite{pavez22}. We say that $(\Mm, \: \Qm)$-GFT has the spectral folding property if for any $(\uv, \lambda)$ generalized eigenpair of the $(\Mm, \Qm)$-GFT, $(\Jm\uv, \: (2 - \lambda))$ is also a generalized eigenpair of the $(\Mm, \Qm)$-GFT (i.e., $\Mm\uv = \lambda \Qm \uv \iff \Mm\Jm\uv = (2 - \lambda)\Qm\Jm\uv$). 
Here, $\Sc=\{1,\hdots,|\Sc|\}$, $\Sc^c =\Vc \setminus \Sc$ is a partition, and $\Jm$ is a diagonal matrix with diagonal entries given by $\Jm_{i,i}=1 \text{\normalfont{ if } } i \in \Sc$ and $\Jm_{i,i}=-1 \text{\normalfont{ if } } i \in \Scc$.
 Furthermore, the $(\Mm, \Qm(\Sc))$-GFT has the spectral folding property \emph{iff} the inner product $\Qm(\Sc)$ is chosen as\footnote{We use the notation $\Qm(\Sc)$ to emphasize that $\Qm$ is a function of $\Sc$.}
    \begin{equation}\label{eq:def_Q}
        \Qm(\Sc) = \begin{bmatrix}
            \Mm_{\Sc\Sc} & \mathbf{0} \\
            \mathbf{0} & \Mm_{\Sc^c\Sc^c}
        \end{bmatrix}.
    \end{equation}
Note that this $(\Mm, \Qm(\Sc))$-GFT is a function of $\Sc$, meaning the subspace for signal interpolation depends on the sampling set. Furthermore, leveraging the spectral folding property and letting $K=|\Sc|$, the reconstruction formulae in~\eqref{eq:mqrec} simplify to~\cite{Paki24}
\begin{align}
    \tilde{\xv} & = 2\Um_{\Vc\Kc}\Um_{\Sc\Kc}\tr\Qm_\Sc(\Sc)\xv_\Sc. \label{eq:sfmqrec}
\end{align}
Note that the $\Qm(\Sc)$-norm of the signal takes the relative contribution of the sampled nodes for the reconstruction and relative reliabilities of interpolation in non-sampled nodes into account~\cite{Paki24}.

\section{Sampling set selection}
\label{sec:repsam}


\subsection{Sampling set selection}
\label{sec:sampling}
Let $\xv$ be a random signal drawn from a distribution with covariance $\bs\Sigma$, with $\Mm=\bs \Sigma^{-1}$, and let $\Qm(\Sc)$ be selected according to \eqref{eq:def_Q}. 
We formulate the sampling set selection problem as choosing the subset that minimizes the worst-case bandlimited model mismatch error (E-optimal) of the signal with respect to the $(\Mm, \Qm(\Sc))$-GFT interpolator with SF introduced in \cite{Paki24}.

We first define three frequency domain subspaces associated with the SF $(\Mm,\Qm(\Sc))$-GFT. 
We define $\Lscr = \text{span}\{ \uv_i : 0 \leq \lambda_i <1 \}$, $\Mscr = \text{span}\{ \uv_i : \lambda_i = 1 \}$,  and $\Hscr =\text{span}\{ \uv_i : 1 <\lambda_i \leq 2 \}$, for low, middle and high frequency signals, respectively. The index sets associated to each sub-space are $\Lc = \{ i : 0 \leq \lambda_i<1 \}$, $\Mc = \{ i : \lambda_i=1 \}$, and $\Hc = \{ i: 1<\lambda_i\leq 2 \}$. 
Without loss of generality, we assume that $|\Sc| < |\Scc|$ and that $\Mm_{\Sc\Scc}$ has full row rank, i.e., $\text{rank}(\Mm_{\Sc\Scc}) = |\Sc|$. It can be proved that $\text{dim}(\Lscr)=\text{dim}(\Hscr) = |\Sc|$ and $\text{dim}(\Mscr)=|\Scc| - |\Sc|$~\cite{pavez22}. Note that these subspaces and frequencies are functions of $\Sc$.

Any signal $\xv$ can be decomposed as $\xv = \xv_l(\Sc) + \Delta\xv_m(\Sc) + \Delta\xv_h(\Sc)$, where $\xv_l(\Sc) \in \Lscr$ is the bandlimited component, $\Delta\xv_m(\Sc) \in \Mscr$, and $\Delta\xv_h(\Sc) \in \Hscr$. Here, $\Delta\xv_m(\Sc) + \Delta\xv_h(\Sc)$ represents the bandlimited model mismatch error of the signal. Furthermore, let $\hat\xv(\Sc) = \begin{bmatrix}\hat\xv_l(\Sc)\tr & \Delta\hat\xv_m(\Sc)\tr & \Delta\hat\xv_h(\Sc)\tr \end{bmatrix}\tr$, where $\Delta\hat\xv_l(\Sc) = \Um_{\Vc\Lc}^\top \Qm(\Sc) \xv$, $\Delta\hat\xv_m(\Sc) = \Um_{\Vc\Mc}^\top \Qm(\Sc) \xv$, and $\Delta\hat\xv_h(\Sc) = \Um_{\Vc\Hc}^\top \Qm(\Sc) \xv$. 

\begin{theorem}
    If we use~\eqref{eq:sfmqrec} to interpolate $\xv_\Scc$ from $\xv_\Sc$, then the reconstruction error in $\Qm(\Sc)$-norm is given by $\|\tilde\xv-\xv\|_{\Qm(\Sc)}^2 = 2\|\Delta\hat\xv_h(\Sc)\|_2^2+\|\Delta\hat\xv_m(\Sc)\|_2^2.$
\end{theorem}
\begin{proof}
\begin{align}
    \|\tilde\xv-\xv\|_{\Qm(\Sc)}^2 &= \|\Um_{\Vc\Lc}\Um_{\Vc\Lc}\tr\Qm(\Sc)(\xv+\Jm\xv) -  \xv\|_{\Qm(\Sc)}^2 \notag \\
    &=\left\Vert \begin{bmatrix} \Delta\hat\xv_l(\Sc)+\Delta\hat\xv_h(\Sc) \\ \bs 0 \end{bmatrix} - \hat\xv(\Sc) \right\Vert_2^2 \notag \\
    &= \left\Vert \begin{bmatrix} \Delta\hat\xv_h(\Sc) \\ -\Delta\hat\xv_m(\Sc) \\ -\Delta\hat\xv_h(\Sc) \end{bmatrix}\right\Vert_2^2 \notag \\
    &= 2\|\Delta\hat\xv_h(\Sc)\|_2^2+\|\Delta\hat\xv_m(\Sc)\|_2^2.
\end{align}
The second equality follows from the generalized Parseval's identity~\cite{girault18}, and the equality $\Um\tr\Qm\Jm\xv=\Um\tr\Qm\Jm\Um\hat\xv = \Rm\hat\xv$ which follows from  $\Um\tr\Qm(\Sc)\Jm\Um = \Rm$.  $\Rm$ is the anti-diagonal permutation matrix~\cite{Paki24}. 
\end{proof}

Note that this property of folding higher-frequency energy into low-frequency (aliasing) while sampling is due to the SF property. Other GFTs spread the error across all frequencies while sampling, which makes it more challenging to isolate and minimize those errors~\cite{Paki24}. The above result states that the reconstruction error with the interpolator in ~\eqref{eq:sfmqrec} depends only on the higher frequency component of the signal. Therefore, minimizing the energy in high-frequency as a function of $\Sc$ effectively optimizes the frequencies for better energy compaction, which is favourable for the reconstruction in~\eqref{eq:sfmqrec}.
\begin{theorem}
    The covariance of the high-frequency transform coefficient $\Delta\hat\xv_h(\Sc)$ is a diagonal matrix given by $\text{cov}(\Delta\hat\xv_h(\Sc)) = \text{diag}\left( \frac{1}{\lambda_N},  \hdots, \frac{1}{\lambda_{N-|\Sc|+1}} \right)$.
\end{theorem}
\begin{proof}
\begin{align}
    \text{cov}(\Delta\hat\xv_h(\Sc)) \notag
    &= \Ex[\Um_{\Vc\Hc}^\top \Qm \xv \xv \tr\Qm\Um_{\Vc\Hc}] \\ \notag
    &= \Um_{\Vc\Hc}^\top \Qm(\Sc) \Ex[\xv \xv \tr]\Qm(\Sc)\Um_{\Vc\Hc} \\ \notag
    &= \Um_{\Vc\Hc}^\top \Qm(\Sc) \Mm^{-1}\Qm(\Sc)\Um_{\Vc\Hc} \\ \notag
    &= \Um_{\Vc\Hc}^\top \Qm(\Sc) \Um \bs\Lambda^{-1}\Um\tr\Qm(\Sc) \Um_{\Vc\Hc} \\ 
    &= \text{diag}\left( \frac{1}{\lambda_N},  \hdots, \frac{1}{\lambda_{N-|\Sc|+1}} \right).
\end{align}
\noindent The last equality follows from $\Um\tr\Qm(\Sc)\Um=\Id$. 
\end{proof}
It is worth noting that the covariance matrix simplifies to a diagonal matrix by choosing $\Qm(\Sc)$  as in~\eqref{eq:def_Q}. Furthermore, this allows us to compare the worst-case high frequency error corresponding to different subspaces generated by different sampling sets of the same size in terms of $\lambda_{N-|\Sc|+1}$.
\begin{rem}
   e can also show that $\Delta\hat\xv_m(\Sc)$ and $\Delta\hat\xv_h(\Sc)$ are uncorrelated, and $\text{cov}(\Delta\hat\xv_m(\Sc)) = \Id$ is an $(|\Scc| - |\Sc|)$-dimensional identity matrix. When the sampling set size is fixed, as is common in sampling set selection methods, $\text{cov}(\Delta\hat\xv_m(\Sc))$ becomes a constant, thus we can omit  $\Delta\hat\xv_m(\Sc)$ from the sampling set selection criteria.
\end{rem}
\begin{prop} The E-optimal sampling design minimizing the maximum singular value of $\text{cov}(\Delta\hat\xv_h(\Sc))$ is given by $\Sc^{\text{opt}} =  \text{arg }\underset{|\Sc|=s}{\text{min }} \: \lambda_{|\Sc|}.$
    
\end{prop}
\begin{proof}
\begin{align}
    \Sc^{\text{opt}} & = \text{arg }\underset{|\Sc|=s}{\text{min }} \sigma_{\text{max}}(\text{cov}(\Delta\hat\xv_h(\Sc)))   \notag \\ & = \text{arg }\underset{|\Sc|=s}{\text{min }} \: \frac{1}{\lambda_{N-|\Sc|+1}} = \text{arg }\underset{|\Sc|=s}{\text{min }} \: \lambda_{|\Sc|}    \label{eq:samp_obj1}
\end{align}
The last equality follows from the spectral folding property, where $\lambda_{N-|\Sc|+1} = 2 - \lambda_{|\Sc|}$, with $\lambda_{|\Sc|}$ being the largest frequency of the SF $(\Mm, \Qm(\Sc))$-GFT that is smaller than $\lambda = 1$. 
\end{proof}
\vspace{-1ex}
\subsection{Relation to GMRF model}
\vspace{-1ex}
Assume $\xv\sim\Nc(\bs 0, \bs\Sigma)$. It is known that the minimum-mean square error (MMSE) estimation of the GMRF is given by $\xv_\Scc^\text{MMSE} = \bs\Sigma_{\Scc\Sc}\bs\Sigma_\Sc^{-1}\xv_\Sc$. The following theorem establishes that the optimal MMSE estimation in a GMRF can be interpreted as a subspace reconstruction with subspace adaptive frequency definition\footnote{\label{note:Q}In this subsection, we use $\Qm = \Qm(\Sc)$ for ease of notation.}:
\begin{theorem}
    \label{th:theorem1}
    In a GMRF, if we assume $\Mm=\bs\Sigma^{-1}$ and $\Qm$ is selected according to \eqref{eq:def_Q}, then $\xv_\Scc^\text{MMSE} = 2\Um_{\Sc\Lc}(\Id-\bs\Lambda_\Lc)\Um_{\Scc\Lc}\tr\Qm_\Scc\xv_\Scc$.
\end{theorem}

From the above theorem, we can infer that the MMSE estimator is a sampling set-adaptive subspace reconstruction. The only difference between MMSE estimator and the interpolator in~\eqref{eq:sfmqrec} is the term $(\Id-\bs\Lambda_\Lc)$. The diagonal entries of $(\Id-\bs\Lambda_\Lc)$ could be viewed as penalties on each of these frequencies. These bases and penalties come from the GMRF model as a function of $\Sc$. 
In the presence of model mismatches, the reliability of these subspaces and frequencies from the graph model becomes questionable. Because of this uncertainty, the interpolator in \eqref{eq:sfmqrec} does not penalize these frequencies and uses the same sampling set-adaptive subspace as the MMSE estimator. This intuition explains why the interpolator in \cite{Paki24} could work better for non-GMRF data in the presence of model mismatch. 

\begin{proof}[Proof of \cref{th:theorem1}]
    Consider the following block-inverse formula:
    \begin{equation} \bs\Mm = \bs\Sigma^{-1}=\begin{bmatrix} \bs\Sigma^{-1}_{\Sc|\Scc}& -\bs\Sigma^{-1}_{\Sc|\Scc}\bs\Sigma_{\Sc\Scc}\bs\Sigma_{\Scc}^{-1}\\ -\bs\Sigma_{\Scc|\Sc}^{-1}\bs\Sigma_{\Scc\Sc}\bs\Sigma_{\Sc}^{-1}& \bs\Sigma_{\Scc|\Sc}^{-1} \end{bmatrix},
    \label{eq:blk_inv}
    \end{equation}
    where $\bs\Sigma_{\Sc|\Scc}=\bs\Sigma_{\Sc}-\bs\Sigma_{\Sc\Scc}\bs\Sigma_{\Scc}^{-1}\bs\Sigma_{\Scc\Sc}$, and similarly for $\bs\Sigma_{\Scc|\Sc}$.  We can observe that $\bs\Sigma_{\Scc\Sc}\bs\Sigma_\Sc^{-1} = -\Qm_\Scc^{-1}\Mm_{\Scc\Sc}$. Then, starting from the identity $\Qm^{-1}\Mm = \Um\bs\Lambda\Um\tr\Qm$, consider the block-wise expansion of $\Id - \Qm^{-1}\Mm = \Um(\Id - \bs\Lambda)\Um\tr\Qm$. We can show that $-\Qm_\Sc^{-1}\Mm_{\Sc\Scc} = \Um_{\Sc\Ic}(\Id - \bs\Lambda)\Um_{\Scc\Ic}\tr\Qm_\Scc$, where $\Ic = \Lc \cup \Mc \cup \Hc$. Then, using the spectral folding properties $\Id - \bs\Lambda_\Lc = -\Rm(\Id - \bs\Lambda_\Hc)\Rm$, $\Um_{\Sc\Hc} = \Um_{\Sc\Lc}\Rm$, and $\Um_{\Scc\Hc} = -\Um_{\Scc\Lc}\Rm$, we can further simplify the above expression as follows:
    \begin{align}
        -\Qm_\Sc^{-1}\Mm_{\Sc\Scc} & = \Um_{\Sc\Lc}(\Id-\bs\Lambda_\Lc)\Um_{\Scc\Lc}\tr\Qm_\Scc+ \notag \\ & \qquad \qquad \Um_{\Sc\Hc}(\Id-\bs\Lambda_\Hc)\Um_{\Scc\Hc}\tr\Qm_\Scc \notag \\ 
        &= 2\Um_{\Sc\Lc}(\Id-\bs\Lambda_\Lc)\Um_{\Scc\Lc}\tr\Qm_\Scc. \label{eq:svd}
    \end{align}
 \end{proof}
\subsection{Representative sampling partitions}
\vspace{-1ex}
Given a sensor network with vertices $\Vc$, representative sampling subset partitioning involves finding a partition $\Pc = \{\Sc_0, \Sc_1, \dots, \Sc_{p-1}\}$, where $\cup_{i=0}^{p-1} \Sc_i = \Vc$, such that each subset $\Sc_i$ serves as a sampling set to reconstruct the signal in its complement $\Sc_i^c$. 
Based on~\eqref{eq:samp_obj1} and following the algorithm proposed in \cite{Chakraborty23}, the algorithm we use for partitioning is outlined in ~\cref{alg:part_algo}.

    \begin{algorithm}[t]
    \begin{algorithmic}[1]
    \caption{Greedy representative sampling subset partitioning} \label{alg:part_algo}
    \renewcommand{\algorithmicrequire}{\textbf{Input:}}
     \renewcommand{\algorithmicensure}{\textbf{Output:}}
     \REQUIRE $\Mm, p$
            \FOR{$i = 1$ to $N$}
                \STATE $m=i \text{ mod } p$
                \STATE $q^* = \text{arg } \underset{q \in \Vc\backslash\cup_{j=0}^{p-1}\Sc_j}{\text{min}} \: \lambda_{|\Sc_m \:\cup \: \{q\}|} $
                \STATE $\Sc_m$ $\gets$ $\: \Sc_m \:\cup \: \{ q^* \}$
            \ENDFOR
    \ENSURE  $\Pc = \{\Sc_0, \Sc_1, \dots, \Sc_{p-1}\}$
    \end{algorithmic}  
    \end{algorithm}
\vskip -1ex

\subsection{Efficient implementation}
\label{sec:approximation}
To reduce the computations associated with the inversion of $\Qm_\Sc$, we use the following Neumann series expansion~\cite{zhang2011matrix}:
\begin{equation}
    \Qm_{\Sc}^{-1} = \Dm_\Sc^{-1/2}\left(\sum_{k=0}^\infty(\Dm_\Sc^{-1/2}\Wm_\Sc\Dm^{-1/2}_\Sc)^k\right)\Dm^{-1/2}_\Sc, 
\end{equation}
where $\Dm$ is a diagonal matrix with $\Dm_{i,i}=\Mm_{i,i}$, and $\Wm=\Dm-\Mm$ is non negative with zero diagonal. The zero-order approximation of $\Qm_\Sc^{-1}$ is given by $\Qm_\Sc^{-1} \approx \Dm^{-1}_\Sc$, and similarly, $\Qm_\Scc^{-1} \approx \Dm^{-1}_\Scc$. Using ~\eqref{eq:svd}, we can show that $\sigma_{\text{min}}(\Qm_\Sc^{-1/2}\Mm_{\Sc\Scc}\Qm_\Scc^{-1/2}) = 1-\lambda_{|\Sc|}$.  Then, the approximated version \eqref{eq:samp_obj1} is given as follows:
\begin{align}
    \Sc^{\text{opt}} & = \text{arg }\underset{|\Sc|=s}{\text{max }} \sigma_{\text{min}}(\Dm_\Sc^{-1/2}\Mm_{\Sc\Scc}\Dm_\Scc^{-1/2}).    \label{eq:samp_obj3}
\end{align}
In step $3$ of ~\cref{alg:part_algo}, we solve \eqref{eq:samp_obj3} instead of \eqref{eq:samp_obj1} for efficient computations.We solve this problem using a greedy approach, similar to other greedy sampling methods in the GSP literature~\cite{tanaka20,chen15}.
\vspace{-1ex}
\section{EMPIRICAL RESULTS}
\vspace{-1ex}
\label{sec:experiments}
We use a synthetic example of a sensor network capturing spatial Gaussian processes to demonstrate the efficiency of the proposed partitioning method. We compare  against~\cite[Algorithm 4]{Chakraborty23}, which optimizes the partition under the assumption of a bandlimited model for signals on graphs, as it is more comparable to our approach. We use the following error metric from \cite{Chakraborty23} to compare the partitions obtained from different methods:
\begin{equation}
    Err(\Pc) = \frac{1}{p}\sum_{j=0}^{p-1}\left(\frac{1}{M}\sum_{i=1}^M \frac{\|\tilde\xv_{\Sc_j^c}^i-\xv_{\Sc_j^c}^i\|_2^2}{\|\xv^i\|_2^2}\right), \label{eq:err_metric}
\end{equation}
where $M$ is the number of signals used to evaluate the performance.

    \begin{figure} 
        \centering
        \subfloat[\label{fig:gp_sensors}]{%
       \includegraphics[width = 0.26 \textwidth, trim={0cm 0cm 0cm 0.5cm}, clip ]{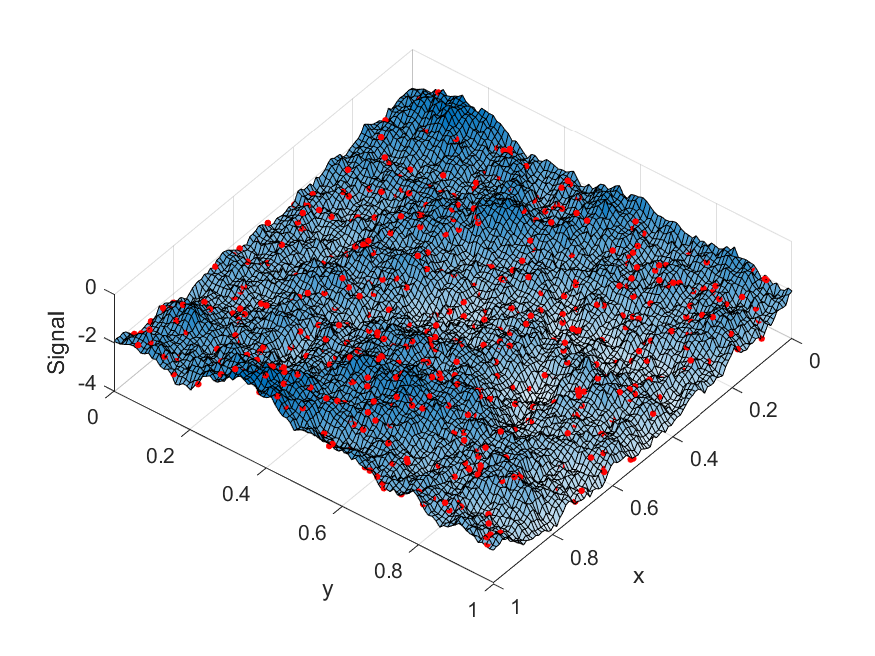}}
        \hfill
      \subfloat[\label{fig:graph_r1}]{%
            \includegraphics[width = 0.22 \textwidth, trim={0cm 0cm 0cm 0.1cm}, clip ]{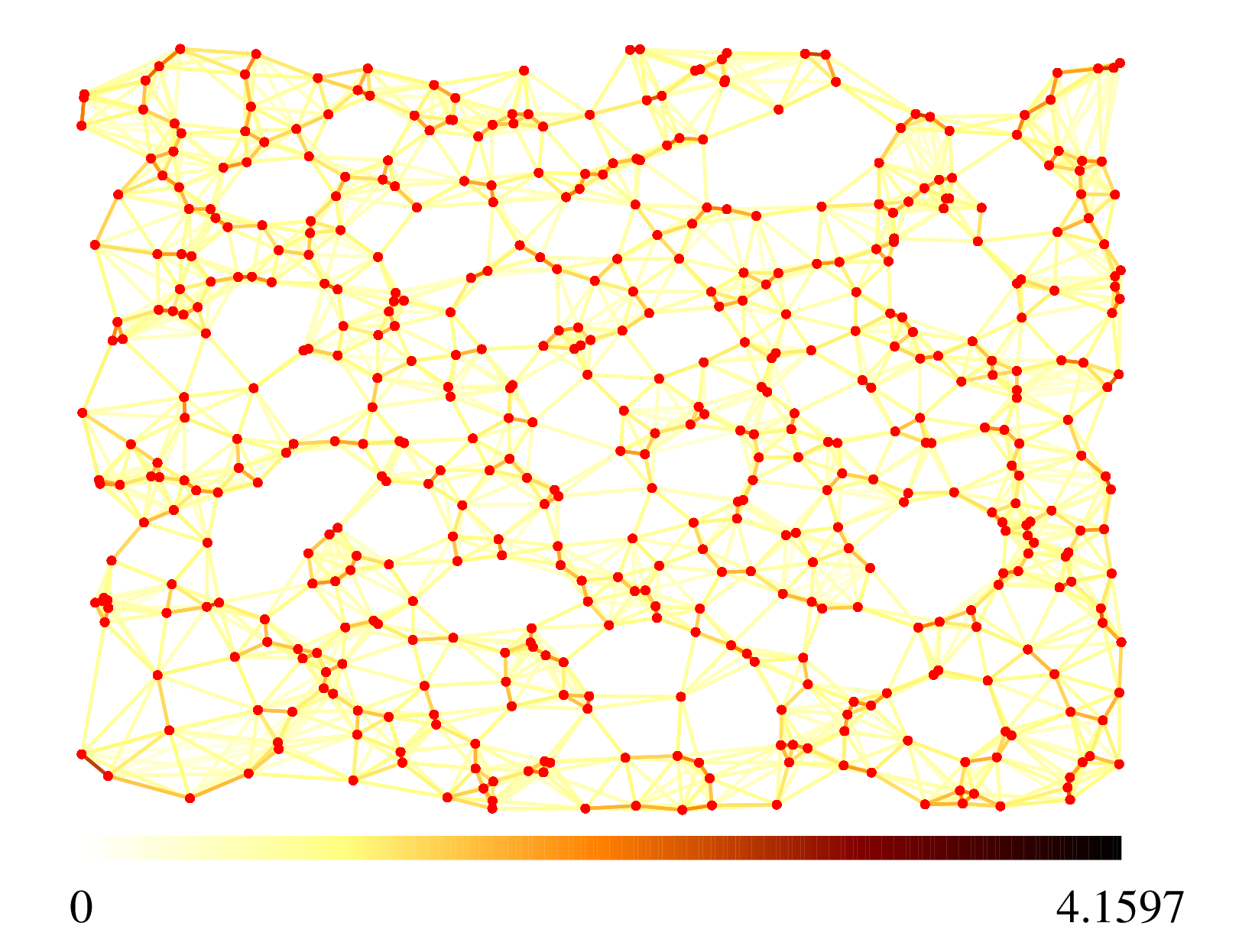}}
        \vskip -2ex
        \caption{(a) An instance of the synthetic Gaussian process with the sensor locations marked in red; (b) The sensor network graph, learned from the synthetic data corresponding to $\sigma=1$.}        
        \vskip -2.5ex
        
    \end{figure}
We consider a dense grid within a $1 \times 1$ square in $\mathbb{R}^2$ and consider Gaussian process with covariance $\Sm$ to model a continuous process in space, where $\Sm_{i,j} = \exp(-d_{i,j}/\sigma^2)$. Here, $d_{i,j}$ represents the Euclidean distance between points $i$ and $j$, and $\sigma$ is a parameter controlling the smoothness of the signal. 
Then, we select $N = 500$ random locations within the square to simulate the sampling of the Gaussian process using sensors. These sensor locations are fixed throughout all the experiments in this section. A realization  of the Gaussian process and sensor locations are shown in~\cref{fig:gp_sensors}. We denote the signal to be measured by $\xv$. We generate $5000$ independent identically distributed (i.i.d.) realizations of $\xv$ and obtain its sample covariance matrix $\bs\Sigma$. The graph for the sensor network is a  combinatorial graph Laplacian learned from this covariance matrix using~\cite[Algorithm 2]{Egilmez17}. This graph is  constrained so that  the sensors are connected only to their neighbours within a radius of $r = 0.3$. The sensor network graph learned for the smoothness parameter $\sigma=1$ (smooth variation in space) is shown in ~\cref{fig:graph_r1}.   

The method in~\cite{Chakraborty23} requires the bandwidth of the signals as an input for the algorithm. To estimate the bandwidth for the bandlimited model with respect to $(\Mm, \Id)$-GFT, we use the same signals employed for learning the graph. First, we generate a partition with $p = 5$ using \cite[Algorithm 4]{Chakraborty23}, assuming the bandwidth is $80$. Then, we reconstruct all the signals in the training set using each subset in the partition for various bandwidths and compute the error metric given in ~\eqref{eq:err_metric}. The plot of $Err$ versus bandwidth for the data corresponding to $\sigma = 1$ is shown in \cref{fig:bw_est}. The optimal bandwidth, $K_{\text{opt}}$, for reconstruction is identified as the bandwidth that minimizes the reconstruction error~\cite{jayawant23}. 
Furthermore, we observed the frequencies $\lambda$ obtained by solving the exact problem in ~\eqref{eq:samp_obj1} and its approximation in ~\eqref{eq:samp_obj3} for $|\Sc|=100$, as shown in~\cref{fig:exact-approx}. This shows that our approximation is reasonably accurate for solving~\eqref{eq:samp_obj1}.
    \begin{figure}
        \centering         
        \subfloat[\label{fig:bw_est}]{\includegraphics[width = 0.22 \textwidth, trim={0cm 0cm 0cm 0.5cm}, clip ]{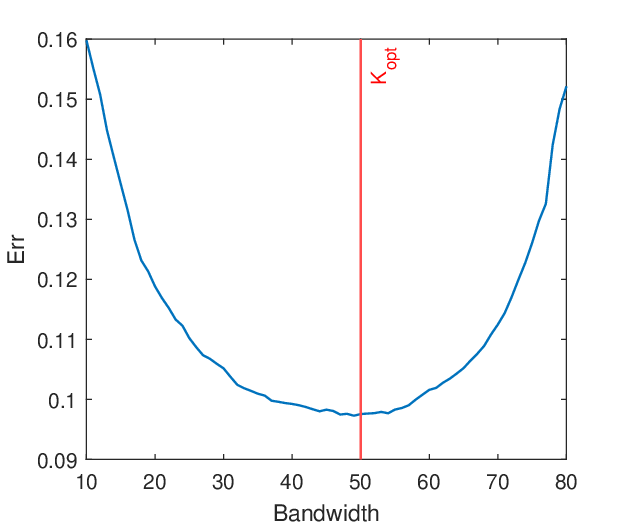}}
        \hfill
        \subfloat[\label{fig:exact-approx}]{\includegraphics[width = 0.22 \textwidth, trim={0cm 0cm 0cm 0.5cm}, clip ]{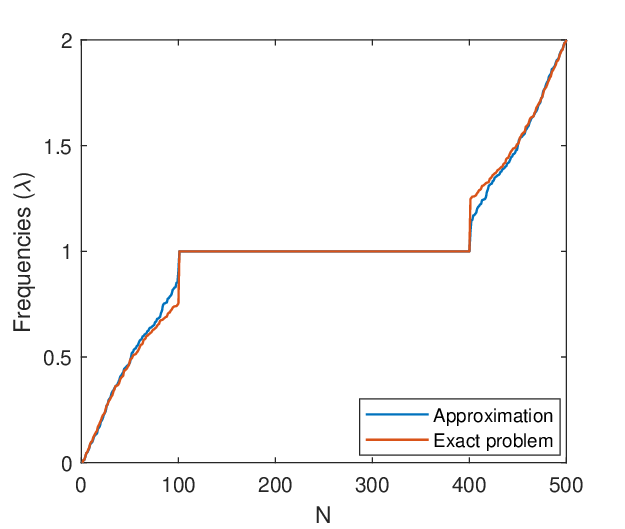}}
        \vskip -1.5ex
        \caption{ (a) The variation $Err$ for $(\Mm, \Id)$-GFT bandlimited reconstruction for different bandwidths. The red line indicates $K_{\text{opt}}$ (b) The frequencies $\lambda$ obtained by solving~\eqref{eq:samp_obj1} and its approximation in~\eqref{eq:samp_obj3}.}
        \vskip -1.5ex
    \end{figure}
    
We chose $p$ such that $p\leq N/K_{\text{opt}}$, to ensure that all subsets $\Sc_i$ can achieve near-optimal reconstruction with $(\Mm, \Id)$-GFT bandlimited reconstruction. We obtain the representative sampling partitions using ~\cref{alg:part_algo} and \cite[Algorithm 4]{Chakraborty23}, for different values of $\sigma$, and $p$. The $K_{\text{opt}}$ value obtained for $(\Mm, \Id)$-GFT bandlimited reconstruction for $\sigma=1$ and $\sigma=0.4$ are $50$ and $46$ respectively. 
The average SNR obtained over $M=500$ test signals for each case with these partitioning and reconstruction methods are shown in \cref{tab:rec_err}. From the table, we observe that for smooth signals with $p=10$, the proposed method outperforms the competing method by 2 dB. In a more challenging scenario, where the signal varies rapidly in space ($\sigma = 0.4$) and $p=10$, our method exceeds the competitor by approximately 4 dB. Furthermore, the small gain observed in the last column when using different partitions and SF $(\Mm, \Qm(\Sc_i))$-GFT reconstruction is due to the robustness of SF $(\Mm, \Qm(\Sc))$-GFT interpolation for arbitrary sampling sets~\cite{Paki24}.

\begin{table}[]
    \centering
    \renewcommand*{\arraystretch}{1.1}
    \begin{tabular}{c c c c c}    
        \hline \smallskip
        \multirow{2}{*}{$\sigma$} & \multirow{2}{*}{$p$} & \multirow{2}{*}{Method}  & $(\Mm, \Id)$-GFT & SF $(\Mm, \Qm(\Sc_i))$-GFT\\
        & & & $Err$ & $Err$ \\
        \hline \smallskip 
         \multirow{4}{*}{1} & \multirow{2}{*}{5} & \cite[Alg.4]{Chakraborty23} & $10.80$   & $11.75$ \\ 
         & & Proposed & $10.29$   & $\bs{11.97}$ \\ 
         \cline{2-5}\smallskip
         & \multirow{2}{*}{10} &  \cite[Alg.4]{Chakraborty23} & $7.75$   & $9.71$ \\ 
         & & Proposed & $7.02$   & $\bs{9.86}$ \\

         \hline \smallskip
         \multirow{4}{*}{0.4} & \multirow{2}{*}{5} & \cite[Alg.4]{Chakraborty23} & $4.81$   & $5.76$ \\ 
         & & Proposed & $4.69$   & $\bs{5.90}$ \\ \cline{2-5}\smallskip
         & \multirow{2}{*}{10} &  \cite[Alg.4]{Chakraborty23} & $1.51$   & $3.83$ \\ 
         & & Proposed & $-0.21$   & $\bs{3.99}$ \\
         \hline 
    \end{tabular}
    \vskip -1ex
    \caption{Comparision of average SNR (in dB); Method column indicates the partitioning algorithm used. The best performance for each $(\sigma,p)$ pair is in bold letters.}
    \label{tab:rec_err}
\vskip -2.5ex
\end{table}

\vspace{-1ex}
\section{CONCLUSION}
\vspace{-1ex}
\label{sec:conclusions}
In this work, we proposed a graph sampling method for interpolation using a sampling-adaptive GFT. We showed that the error in reconstruction with the sampling set adaptive interpolator depends only on the high-frequency energy of the signal. Using this, we proposed a sampling algorithm that minimizes the worst-case bandlimited model mismatch error with respect to the sampling set-adaptive GFT.  Our primary application involved partitioning sensor networks into multiple representative subsets, with each subset using its own sampling set-adapted GFT for interpolation. Experimental results showed that this approach outperforms traditional methods based on fixed subspace. Future work could focus on extending this partitioning algorithm for distributed implementation in large-scale sensor networks.

\bibliographystyle{IEEEbib}
\bibliography{BibGSP}
\vfill
\end{document}